\documentclass[draftclsnofoot,onecolumn,12pt]{IEEEtran}

\ifCLASSINFOpdf
\else
\fi

\usepackage{setspace}
\usepackage{bbm}
\usepackage[cmex10]{amsmath}
\usepackage{amssymb}
\usepackage{cite}
\usepackage[dvipdfmx]{graphicx} 
\usepackage{bmpsize}
\usepackage{array,color}
\usepackage{amsmath}
\usepackage{stfloats}
\usepackage{epstopdf}
\usepackage{subfigure}
\usepackage{tabularx}
\usepackage{epsfig,epsf,color,balance,cite}
\usepackage{algorithmic}
\usepackage{algorithm}
\usepackage{url}
\usepackage{bm}
\usepackage{footnote}
\usepackage{amsthm}

\newtheorem{theorem}{Theorem}

\newtheorem{proposition}[theorem]{Proposition}

\ifodd 1
\usepackage{soul}
\usepackage{color}
\setstcolor{red}

\newcommand{\revh}[1]{{\color{black}#1}} 
\newcommand{\del}[1]{\st{#1}} 

\newcommand{\com}[1]{\textbf{\color{red} (COMMENT: #1)}} 
\newcommand{\response}[1]{\textbf{\color{green} (RESPONSE: #1)}} 
\else

\newcommand{\revh}[1]{#1}

\newcommand{\del}[1]{}
\allowdisplaybreaks[4]
\newcommand{\com}[1]{}
\newcommand{\comg}[1]{}
\newcommand{\response}[1]{}
\fi
\allowdisplaybreaks[4]
\newcommand{\tabincell}[2]{\begin{tabular}{@{}#1@{}}#2\end{tabular}}

\title{Reconfigurable Intelligent Surfaces with Reflection Pattern Modulation: 
		Beamforming Design and Performance Analysis}

\author{Shaoe~Lin, Beixiong~Zheng, George~C.~Alexandropoulos, {\it Senior Member, IEEE},\\
	Miaowen~Wen, {\it Senior Member, IEEE}, Marco~Di~Renzo, {\it Fellow, IEEE}, and~Fangjiong~Chen
	\thanks{
		
		S. Lin and F. Chen are with the School of Electronic and Information Engineering, South China University of Technology, Guangzhou 510641, China (e-mail: eeshe.lin@mail.scut.edu.cn, eefjchen@scut.edu.cn).
		
		B. Zheng is with the Department of Electrical and Computer Engineering, National University of Singapore,  Singapore 117583 (e-mail: elezbe@nus.edu.sg).
		
		G. C. Alexandropoulos is with the Department of Informatics and Telecommunications, National and Kapodistrian University of Athens, Panepistimiopolis Ilissia, 15784 Athens, Greece (e-mail: alexandg@di.uoa.gr). 
		
		M. Wen is with the School of Electronic and Information Engineering,
		South China University of Technology, Guangzhou 510640, China, and also
		with the National Mobile Communications Research Laboratory, Southeast
		University, Nanjing 210096, China (e-mail: eemwwen@scut.edu.cn).
		
		M. D. Renzo is with the Université Paris-Saclay, CNRS, CentraleSupélec, Laboratoire des Signaux et 
		Systèmes, Gif-sur-Yvette, France
		(e-mail: marco.direnzo@centralesupelec.fr).
}
}

\begin{document}
\maketitle

\begin{abstract}

Recent considerations for reconfigurable intelligent surfaces (RISs) assume that RISs can convey information by reflection without the need of transmit radio frequency chains, which, however, is a challenging task. 
In this paper, we propose an RIS-enhanced multiple-input single-output system with reflection pattern modulation, where the RIS can configure its reflection state for boosting the received signal power via passive beamforming and simultaneously conveying its own information via reflection. 
We formulate an optimization problem to maximize the {\it average} received signal power by jointly optimizing the active beamforming at the access point (AP) and passive beamforming at the RIS for the case where the RIS’s state information is statistically known by the AP, and propose a high-quality suboptimal solution based on the alternating optimization technique.
	We analyze the asymptotic outage probability of the proposed scheme under Rayleigh fading channels, for which a closed-form expression is derived.
	The achievable rate of the proposed scheme is also investigated for the case where the transmitted symbol is drawn from a finite constellation. 
	Simulation results validate the effectiveness of the proposed scheme and reveal the effect of various system parameters on the achievable rate performance.
	It is shown that the proposed scheme outperforms the conventional RIS-assisted system without information transfer in terms of achievable rate performance.  

\end{abstract}
\IEEEpeerreviewmaketitle

\section{Introduction}
To meet the demanding requirements for fifth generation (5G) wireless communication in, e.g., enhanced data rate, massive connectivity, low latency, ultra reliability, etc., multiple key technologies including millimeter wave (mmWave) communications, massive multiple-input multiple-output (MIMO) systems, and ultra-dense networks (UDNs) have been extensively investigated in the last decade \cite{Boccardi20145G}. 
However, most of those technologies generally require increased implementation complexity and result in considerably increased energy consumption
 \cite{Sohrabi2016Hybrid,Ngo2013Energy}. 
Leveraging the recent advances in reconfigurable metasurfaces \cite{Liu2018Metasurfaces,Tie2014Coding}, reconfigurable intelligent surfaces (RISs) (a.k.a. intelligent reflecting surfaces) have emerged as a revolutionary technology for improving the coverage and energy/spectrum efficiency of future wireless communications \cite{Renzo2020Smart,Tan2018Enabling,Liaskos2018a,basar2019wireless,Huang2019Holographic,Renzo2019Reflection,Nemanja2019Channel}. 
Specifically, RISs are planar metasurfaces consisting of a large number of low-cost unit cell elements, each of which reflects the incident signals according to its reflection state. 
By configuring the reflection amplitude and/or phase shift at each unit cell element according to the dynamic wireless channels, the signals reflected by the RIS can constructively or destructively combine with the signals from other paths at the receiver for signal power enhancement or co-channel interference suppression \cite{wu2017towards}. 
Compared to traditional relays, RISs are envisioned to work in a full-duplex mode without incurring self-interference and thermal noise, and yet possess substantially reduced hardware cost and energy consumption owing to their nearly passive components \cite{Huang2019RIS,Huang2018EE,Emil2019IRS,Wankai2019Wireless}. 

To achieve the theoretical passive beamforming gain offered by RISs in RIS-assisted systems, the global channel state information (CSI) was assumed to be available at the RIS side \cite{Wu2018Intelligent}, which is practically difficult to implement. 
The acquisition/estimation of channels involving RIS is in practice a challenging task due to the absence of signal processing capability of passive RISs and their large number of unit cell elements, which has spurred rapidly growing interests \cite{Tahaa2019enabling,Alexandropoulos2020A,Huang2019indoor,Mishra2019Channel,Z2019Cascaded,Yifei2019Intell,zheng2019Intell,you2019intell}. 
Existing RIS channel acquisition methods can be classified into two categories. 
The first category \cite{Tahaa2019enabling,Alexandropoulos2020A} includes the approaches considering that RISs are embedded with some receive radio frequency (RF) chains, 
which enable them to explicitly estimate the channels from the network's transmitters/receivers to RIS. 
In the second category, instead of explicitly estimating the individual channels from the network's transmitters/receivers to RIS, the cascaded transmitter-RIS-receiver channels are estimated at the receiver by adopting an element-by-element ON/OFF-based reflection pattern at RIS without the need of costly receive RF chains \cite{Mishra2019Channel}. 
However, this approach incurs prohibitive channel estimation overhead due to the large number of RIS elements. To reduce the channel estimation overhead, a grouping strategy for the RIS elements was proposed in \cite{Yifei2019Intell}, where adjacent RIS elements are grouped together to share a common reflection coefficient. Based on this grouping strategy, a channel estimation method using the discrete Fourier transform (DFT)-based reflection pattern was first presented in \cite{zheng2019Intell} for frequency-selective fading channels, which significantly improves the channel estimation accuracy by leveraging the large aperture gain of RIS. In \cite{Tobias2019An}, a DFT-based reflection pattern was designed for channel estimation in RIS-assisted narrowband systems. More recently, some preliminary works (e.g., \cite{Li2020Parallel,Beixiong2020Intelligent}) have appeared for channel estimation in RIS-assisted multi-user systems.

Prior works on RISs mainly focus on their utilization for maximizing the received signal strength of links between the access point (AP) and users\revh{, e.g., \cite{Huang2020RIS}}. 
In order to improve the spectral efficiency of RIS-assisted systems, the authors of \cite{Basar2019Large} introduced the RIS-space shift keying and RIS-spatial modulation schemes to implement the index modulation at the receiver side. 
In addition, the authors of \cite{Roy2019Beyond} proposed to jointly encode the information in the transmitted signal and RIS state for improving spectral efficiency. 
On the other hand, another type of the RIS-assisted system, in which the RIS needs to deliver its own information to the receiver, was considered in \cite{renzo2019smart}. 
In this case, the RIS is equipped with sensors, such as temperature and/or humidity sensors, and
has a requirement on conveying its locally acquired information to an intended receiver. 
Besides the information from sensors, there are other potential sources of RIS information, including the RIS control signaling, RIS maintenance data, estimated CSI at the RIS, etc., \cite{Alexandropoulos2020A,yuan2020ris}. 
One option for the information transfer of the RIS is to install multiple transmit RF chains, which, however, incurs much higher hardware cost and energy consumption as compared to the fully-passive RIS case \cite{Huang2019Holographic}. Alternatively, when no transmit RF chain is installed at the RIS, its own information should be implicitly delivered via appropriately designed reflection state, which is a more cost-effective, yet challenging, solution. 
To this end, a passive beamforming and information transfer (PBIT) scheme was proposed in \cite{yan2019pbit} and \cite{yan2019pbit2} to convey the RIS's information besides implementing fully-passive beamforming. 
However, in this scheme, the number of activated (ON-state) elements\footnote{The ON/OFF switching of the RIS elements can be implemented by the switchers.} varies over time, which can cause significant fluctuation in the reflected signal power, resulting in a relatively high outage probability. 

\begin{figure}[!t]
	\centering
	\includegraphics[width=2.8in]{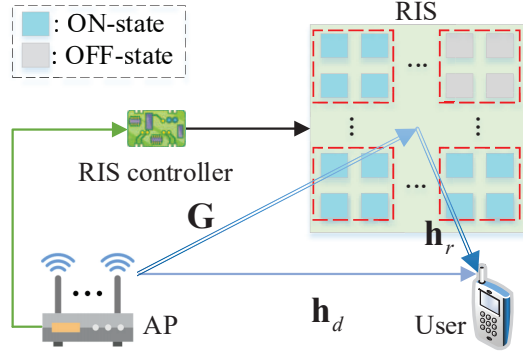}
	\caption{The considered RIS-enhanced downlink MISO wireless communication system. The RIS adopts reflection pattern modulation (RPM) via the combination of its elements' ON/OFF states to simultaneously enhance the communication between the AP and receiving user, and convey its own information data.} 
	\label{system}
\end{figure}

In this paper, we consider an RIS-enhanced multiple-input single-output (MISO) wireless communication system as shown in Fig.~\ref{system}, where a multi-antenna AP communicates with a single-antenna user with the aid of an RIS. 
The RIS is explored as an efficient and inexpensive dual-use technology of passive beamforming and information transfer by proposing the concept of reflection
pattern modulation (RPM). The core RPM idea is to activate a subset of RIS unit cell elements in order 
to reflect a sharp beam towards the intended destination, while exploiting the indices of the activated (ON-state) elements to implicitly convey the RIS's information. 
The main contributions of this paper are summarized as follows:

\begin{itemize}
\item 
We present an RIS-based RPM (RIS-RPM) scheme for PBIT, in which the RIS plays the following two roles: 1)  it boosts the received signal power at the intended receiver; 2) it transfers via reflection its locally acquired information data. 
We consider the practical scenario where the ON/OFF state information of the RIS is statistically known by the AP, for which an optimization problem is formulated to maximize the {\it average} received signal power by jointly optimizing the active beamforming at the AP and passive beamforming at the RIS. 
The formulated problem, however, is non-convex and thus difficult to solve optimally. 
As such, we propose an efficient algorithm based on the alternating optimization technique to find a high-quality suboptimal solution. 
Moreover, we formulate and solve an optimization problem to maximize the {\it instantaneous} received signal power by designing the active and passive beamforming based on the RIS's instantaneous ON/OFF state information, which serves as a performance upper bound for the practical beamforming design based on the RIS's statistical ON/OFF state information.

\item 
The asymptotic outage probability of the proposed RIS-RPM scheme over Rayleigh fading channels is derived  in closed-form.  
It is shown that 
the RIS is able to drastically increase the diversity gain by properly designing the phase shifts of its elements. 
Moreover, we analyze the achievable rate of the RIS-RPM scheme for the practical case where the transmitted symbol is drawn from a finite constellation input. 
Simulation results validate the effectiveness of the proposed  RIS-RPM scheme as well as the proposed optimization algorithm, and unveil the effect of various system parameters on the achievable rate performance. 
It is shown that the RIS-RPM scheme outperforms the conventional RIS-assisted system with full-ON RIS reflection \cite{Wu2018Intelligent} in terms of achievable rate performance, despite the loss in the received signal strength. 
Moreover, the RIS-RPM scheme is shown to have the potential to strike an attractive trade-off between the received signal power and achievable rate performance by varying the number of activated (ON-state) RIS elements.


\end{itemize}

The rest of this paper is organized as follows. In Section~\ref{sm and ce}, we introduce the system model associated with channel acquisition. In Sections~\ref{beamforming2} and \ref{beamforming1}, we formulate two optimization problems for designing active and passive beamforming with efficient solutions. 
Performance analysis in terms of outage probability and achievable rate is presented in Section~\ref{analysis}. 
Simulation results are presented in Section~\ref{simulation} and conclusions are drawn in Section~\ref{conclusion}.

\emph{Notation}: Upper and lower case boldface letters denote matrices and column vectors, respectively.  
${\left(\cdot\right)}^\dagger$,  ${\left(\cdot\right)}^T$, ${\left(\cdot\right)}^H$, and ${\left(\cdot\right)}^{-1}$ represent conjugation,  transpose, Hermitian transpose, and inversion operations, respectively. 
$\left[\bf X\right]_{\iota,\jmath}$ denotes the ($\iota,\jmath$)-th entry of matrix $\bf X$, and $\left[\bf x\right]_{\iota}$ denotes the $\iota$-th entry of vector $\bf x$. 
${\bf I}_N$ denotes an $N\times N$ identity matrix, \revh{and ${\bf 1}_{N}$ and ${\bf 0}_{N}$ denote $N$-dimensional all-one and all-zero column vectors, respectively.} 
$\lVert{\bf x}\rVert_0$ and $\lVert{\bf x}\rVert$ denote the zero norm and Euclidean norm of vector $\bf x$, respectively. 
$\mathbb{C}^{n\times m}$ denotes the set of $n\times m$ complex-valued matrices. 
$\text{diag}(\bf x)$ denotes a diagonal matrix with each diagonal entry being the corresponding entry in vector $\bf x$. 
$\text{tr}({\bf X})$ and $\text{rank}({\bf X})$ denote the trace and rank of matrix $\bf X$, respectively. 
$|\cdot|$ and \revh{$\angle(\cdot)$} denote the modulus and phase of a complex number, respectively.  
$p\left(\cdot\right)$ denotes the probability density function (PDF).
$\left\lceil\cdot\right\rceil$ represents the ceiling function that returns the least integer greater than or equal to the argument. 
${n\choose{k}}$ denotes the binomial coefficient. 
$\mathbb{A}\setminus \mathbb{B}$ is the complement set of $\mathbb{B}$ with respect to $\mathbb{A}$. 
$\mathcal{CN}( {\bm \mu},{\bf \Sigma})$ denotes the distribution
of a circularly symmetric complex Gaussian random vector with covariance matrix $\bf \Sigma$ and mean $\bm \mu$. 
$\mathbb{E}_X\{\cdot\}$ denotes the expectation over random variable $X$. Let ${\bf X}\succeq 0 $ denote that a Hermitian matrix $\bf X$ is positive semi-definite. 

\section{System Description and Channel Estimation}\label{sm and ce}

As illustrated in Fig.~\ref{system}, we consider a MISO wireless communication system, where an RIS composed of $L$ unit cell elements is deployed to enhance the communication link from an AP equipped with $N$ transmit antennas to a single-antenna user. 
An RIS controller is attached
to the RIS, which is responsible for reconfiguring the phase shifts and/or reflection amplitudes of the unit cell elements \cite{Yang2017design} and exchanging information with the AP for the realization of the simultaneous passive beamforming and reflection modulation. 
We assume that the RIS controller is equipped with a number of sensors in order to monitor/collect environmental data, which is typically low-rate bursty data for the user. 
In this paper, we consider quasi-static block fading channels for the AP-user, AP-RIS, and RIS-user links, where all the channels remain unchanged over the coherence block equal to the duration of $T_c$ symbol sampling periods, and are independent and identically distributed (i.i.d.) across different coherence blocks. This assumption is valid since RISs are practically used to support low-mobility users in their vicinity. 

\subsection{System Model}\label{sm}

{\renewcommand{\arraystretch}{1.4}%
\begin{table}[!t]
	\begin{center}
		\caption{An Example of the Proposed RPM Scheme with an RIS of $L=4$ Unit Cell Elements ($K=3$ ON-state Elements). The Index of the OFF-state Element is Used to Implicitly Convey the RIS's Information Bits.}
		\begin{tabular}{|c|c|c|c|c|}
			\hline
			Information bits of RIS & 00 & 01 & 10 & 11 \\
			\hline
			Index of OFF-state element & \{4\}  & \{3\} & \{2\} & \{1\} \\
			\hline
			\tabincell{c}{Reflection state of RIS} & \begin{minipage}{0.085\linewidth}
				\includegraphics[height=2.5\baselineskip]{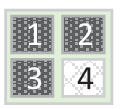}
			\end{minipage} & \begin{minipage}{0.085\linewidth}
			\includegraphics[height=2.5\baselineskip]{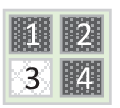}
		\end{minipage} & \begin{minipage}{0.085\linewidth}
		\includegraphics[height=2.5\baselineskip]{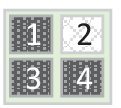}
	\end{minipage} & \begin{minipage}{0.085\linewidth}
	\includegraphics[height=2.5\baselineskip]{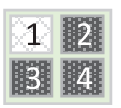}
\end{minipage} \\
			\hline
		\end{tabular}
	\end{center}
\label{lookup}
\end{table}
}

For each symbol duration, only $K$ out of $L$ $(K\le L)$ unit cell elements are turned ON to reflect incident signals on purpose so that the indices of those ON-state elements can implicitly convey the RIS information. 
With the proposed RPM scheme, 
there are in total ${{L}\choose{K}} $ combinations for determining the indices of $K$ ON-state elements at the RIS, such that  $\log_2 {{L}\choose{K}}$ information bits can be conveyed. 
An example of the proposed RPM scheme with an RIS of $L=4$ unit cell elements ($K=3$ ON-state elements) is illustrated in Table~I. 
However, the number of unit cell elements at the RIS can be practically large, leading to an enormous calculation and storage overhead in the combination mapping as the number of $K$-combinations increases exponentially with $L$. 
To address this issue, we adopt the RIS-elements grouping method, 
where a total number of $L$ unit cell elements are divided into $G$ groups, each of which consists of ${\bar L}=L/G$ adjacent elements sharing 
a common reflection coefficient. 
For notational convenience, we assume that ${\bar L}$ is an integer that we refer to as the grouping size. 
As such, the information of the RIS is conveyed through the indices of the ON-state groups. Specifically, for each symbol duration, ${\bar K}$ out of $G$ $(\bar{K}\le G)$ groups are randomly turned ON for reflecting the incident signals, and the remaining $(G-\bar K)$ groups are deliberately turned OFF for realizing the proposed RPM scheme. 
Therefore, 
a total number of $\log_2 {G\choose{\bar K}}$ information bits can be conveyed. 
Let $\mathbb{I}\triangleq\left\{{i_1},{i_2},\ldots,{i_{\bar K}}\right\}$ denote the indices of the $\bar{K}$ ON-state groups, 
where $i_k \in \mathbb{G}\triangleq \left\{1,2,\ldots, { G}\right\}$ for $k=1,2,\ldots,{\bar K}$ with $i_1< i_2<\ldots<i_{\bar K}$. 
Given the indices of the ON-state groups, the common reflection coefficient for the $g$-th group can be characterized by
\begin{align}\label{IRS_group}
\theta_g=\left\{ \begin{gathered}
  \beta_g  e^{-j\phi_g},\quad g \in \mathbb{I}, \hfill \\
0 ,\quad\quad\quad g  \in \mathbb{G}\setminus \mathbb{I}, \hfill
\end{gathered}  \right.\quad g=1,2,\ldots,{G}
\end{align}
where $\beta_g \in [0,1]$ and $\phi_g\in(0, 2\pi]$ represent the common reflection amplitude and phase shift for the $g$-th group, respectively. 
To enhance the reflected signal power and ease the hardware design, the reflection amplitudes of the ON-state groups are set to be the maximum value, i.e., $\beta_g=1$, $\forall g\in\mathbb{I}$. 
Let ${\bm \theta}\triangleq \left[{ \theta}_1, { \theta}_2,\ldots, { \theta}_G   \right]^T$ denote the RIS reflection vector after grouping, which characterizes the equivalent interaction of the RIS with the incident signals. 
According to (\ref{IRS_group}), we have $\|{\bm \theta}\|_0={\bar K}$. 


Let ${\bf G}\triangleq[{\bm \Delta}_{1},{\bm \Delta}_{2},\ldots,{\bm \Delta}_{G}]^H\in\mathbb{C}^{L\times N}$, ${\bf h}_r^H\triangleq[{\bf r}_{1}^H,{\bf r}_{2}^H,\ldots,{\bf r}_{G}^H] \in\mathbb{C}^{1\times L}$, and ${\bf h}_d^H \in\mathbb{C}^{1\times N}$ denote the baseband channels from the AP to RIS, from the RIS to user, and from the AP to user, respectively, where ${\bm \Delta}_{g}^H\in\mathbb{C}^{\bar{L}\times N}$ and ${\bf r}_{g}^H\in\mathbb{C}^{1\times \bar{L}}$ denote the corresponding channels associated with the RIS elements belonging to the $g$-th group, respectively. 
Moreover, due to the severe path loss and high attenuation, {the} signals reflected by the RIS more than once have negligible power and hence can be ignored. 
Accordingly, the received signal at the user 
is given by
\begin{align}
{ y}= \sqrt{P_t}\left( \sum_{g=1}^{G}{\bf r}_{g}^H { \theta}_g {\bm \Delta}_{g}^H  +{\bf h}_d^H \right) {\bf w}{ x}  +{n}\label{received1}
\end{align}
where $P_t$ is the maximum transmit power of the AP, ${\bf w} \in \mathbb{C}^{N \times 1}$ stands for the active beamforming at the AP,
${ x}$ is the transmitted symbol,
which is drawn from an $M\text{-ary}$ constellation $\mathbb{A}$ with normalized power, $\theta_g$ is the common reflection coefficient for the $g$-th group, 
and ${ n}\sim \mathcal{CN}( { 0},\sigma^2)$ is the additive white Gaussian noise (AWGN) with $\sigma^2$ being the noise power. 
Since RIS elements in the same group share a common reflection coefficient, by denoting $\bar{\bf h}_g^H={\bf r}_{g}^H {\bm \Delta}_{g}^H$ as the equivalent cascaded channel of the AP-RIS-user link associated with the $g$-th group without the effect of RIS reflection, (\ref{received1}) can be rewritten as
\begin{align} 
{ y}=\sqrt{P_t}\left({\bm \theta}^T{\bf H}+ {\bf h}_d^H\right) {\bf w}{ x}+{ n}\label{received2}
\end{align}
where ${\bf H}=\left[\bar{\bf h}_{1},\bar{\bf h}_{2},\ldots,\bar{\bf h}_{G}\right]^H\in\mathbb{C}^{G \times N}$ denotes the cascaded AP-RIS-user channel matrix without the effect of RIS reflection. 
From (\ref{received2}), it can be observed that the information delivered to the user consists of two parts:
the first part is from the AP {\it explicitly} expressed as ${ x}$ and the second part is from the RIS  {\it implicitly} embedded in ${\bm \theta}$.

\subsection{Channel Estimation}\label{ce}
	
	To boost the received signal power by jointly optimizing the active beamforming at the AP and passive beamforming at the RIS, the knowledge of $\bf H$ and ${\bf h}_d$ is required. 
	By assuming channel reciprocity and using time-division duplex (TDD) protocol, 
	the CSI of $\bf H$ and ${\bf h}_d$ can be obtained at the AP. 
Specifically, during the channel training, $(G+1)$ consecutive pilot symbols are sent by the user. 
	For the duration of the $i$-th ($i=0,1,\ldots,G$) pilot symbol, the received pilot signal vector at the AP can be expressed as
	\begin{align}
		{\bf y}_p^{(i)}=\sqrt{P_p}\left(\sum_{g=1}^{G} {\bm \Delta}_{g} { \psi}_g^{(i)}  {\bf r}_{g}+{\bf h}_d\right) { x_p^{(i)}}+{\bf n}^{(i)}\label{ul1}
		\end{align}
	where $P_p$ is the transmit power of the user, $x_p^{(i)}$ is the $i$-th pilot symbol with normalized power, ${ \psi}_g^{(i)}$ is the common phase shift for the $g$-th group during the transmission of the $i$-th pilot symbol with $|{ \psi}_g^{(i)}|=1$, and ${\bf n}^{(i)}\sim \mathcal{CN}( {\bf 0}_N,\sigma^2{\bf I}_N)$ is the AWGN vector. 
	By letting ${\bm \psi}^{(i)}\triangleq \left[{ \psi}_1^{(i)}, { \psi}_2^{(i)},\ldots, { \psi}_G^{(i)}   \right]^T$ denote the RIS phase-shift state during the transmission of the $i$-th pilot symbol, (\ref{ul1}) can be rewritten as
		\begin{align}
		{\bf y}_p^{(i)}=\sqrt{P_p}\left({\bf H}^H {\bm \psi}^{(i)}  +{\bf h}_d\right) { x_p^{(i)}}+{\bf n}^{(i)}
		= \sqrt{P_p}\tilde{\bf H}\tilde{\bm \psi}^{(i)}{ x_p^{(i)}}+{\bf n}^{(i)}
		\end{align}
		where $\tilde{\bf H}=\left[{\bf h}_d\quad{\bf H}^H\right]\in\mathbb{C}^{N \times (G+1)}$ and $\tilde{\bm \psi}^{(i)}=\left[ {\begin{array}{*{20}{c}}
			{1}\\
			{{\bm \psi}^{(i)}}
			\end{array}} \right]$. 
		By stacking $(G+1)$ consecutive received pilot signals, we have
		\begin{align}
		{\bf Y}_p=\left[{\bf y}_p^{(0)},{\bf y}_p^{(1)},\ldots,{\bf y}_p^{(G)}\right]=\sqrt{P_p}\tilde{{\bf H}} {{\bm \Psi}}\text{diag}\left( {\bf x}_p \right) +{\bf \Omega}\label{CE}
		\end{align}
		where ${\bm \Psi}=\left[\tilde{\bm \psi}^{(0)},\tilde{\bm \psi}^{(1)},\ldots,\tilde{\bm \psi}^{(G)}\right]$ is the RIS phase-shift pattern, ${\bf x}_p=\left[x_p^{(0)},x_p^{(1)},\ldots,x_p^{(G)}\right]^T$ denotes the pilot sequence, and ${\bf \Omega}=\left[{\bf n}^{(0)},{\bf n}^{(1)},\ldots,{\bf n}^{(G)}\right]$ is the AWGN matrix. 
		By using the DFT-based phase-shift pattern \cite{Tobias2019An}, the CSI of $\bf H$ and ${\bf h}_d$ can be estimated as
		\begin{align}\label{estimate}
		\hat{\tilde{{\bf H}}}=\left[\hat{{\bf h}}_d\quad\hat{{\bf H}}^H\right]=\frac{1}{\sqrt{P_p}(G+1)}{\bf Y}_p \text{diag}\left( {\bf x}_p \right)^{-1}{\bf F}_{G+1}^H
		\end{align}
		where $\left[{\bf F}_{G+1}\right]_{\iota,\jmath}=e^{-j\frac{2\pi\iota\jmath}{G+1}}$, $\iota,\jmath=0,1,\ldots,G$. 
			
			After acquiring the CSI of $\bf H$ and ${\bf h}_d$, the AP jointly optimizes the active beamforming $\bf w$ and phase shifts of all RIS-elements groups to boost the received signal power, 
			and then informs the RIS controller the optimized phase shifts to be implemented at the corresponding unit cell elements. 
			In the following section, we formulate an optimization problem to maximize the {\it average} received signal power by jointly designing the active beamforming at the AP and passive beamforming at the RIS based on the statistical ON/OFF state information of the RIS. 

\section{Beamforming Design Based on Statistical ON/OFF State Information}\label{beamforming2}

\subsection{Problem Formulation}

We consider the practical assumption that the ON/OFF state information of the RIS is statistically known by the AP. 
	Our objective is to minimize the  outage probability of the combined AP-user channel
	by jointly optimizing the active beamforming at the AP and passive beamforming at the RIS under the constraints of the maximum transmit power  of the AP and unit-modulus reflection of the RIS.  
Specifically, given the indices of the ON-state groups ${\mathbb I}$, let ${\bf s}\triangleq \left[{ s}_1, { s}_2,\ldots ,{ s}_G \right]^T$ denote the ON/OFF states of all  RIS-elements groups in a vector form, with each entry given by
\begin{align}\label{state}
{s}_g=\left\{ \begin{gathered}
1,\quad\quad g \in {\mathbb I}, \hfill \\
0,\quad\quad g \in \mathbb{G}\setminus \mathbb{I}, \hfill
\end{gathered}  \right. \quad 
\end{align}
where ${s}_g=1$ and ${s}_g=0$ represent the ON and OFF states, respectively. 
Based on (\ref{state}), the RIS reflection vector ${\bm \theta}$ can be rewritten as ${\bm \theta}={\bm \Phi}{\bf s}$, in which ${\bm \Phi}=\text{diag}\left({\bm \varphi} \right)$ is the diagonal phase-shift matrix of
the RIS with the $g$-th diagonal entry given by $\varphi_g= e^{-j\phi_g}$, $\phi_g\in(0, 2\pi]$, $g=\{1,2,\ldots,G\}$. 
The combined channel from the AP to user is given as ${\bm \theta}^T{\bf H}+ {\bf h}_d^H={\bf s}^T {\bm \Phi} {{\bf H}} + {{\bf h}}_d^H$, and thus the
	achievable rate of the combined AP-user channel conditioned on ${{\bf H}}$ and ${{\bf h}}_d$ can be expressed as
	$
	\log_2 \left( 1+ \frac{P_t}{\sigma^2}
	\left|\left({\bf s}^T {\bm \Phi} {{\bf H}} + {{\bf h}}_d^H\right) {\bf w}\right|^2 \right) 
	$ bits/second/Hertz.
	The corresponding outage probability for a fixed rate $R$ is given by
	\begin{align}\notag
	p_\text{out} (R) &=\mathbb{P} \left\{ \log_2 \left( 1+ \frac{P_t}{\sigma^2}
	\left|\left({\bf s}^T {\bm \Phi} {{\bf H}} + {{\bf h}}_d^H\right) {\bf w}\right|^2 \right) < R \right\}\\\label{out}
	& = \mathbb{P} \left\{ {P_t}
	\left|\left({\bf s}^T {\bm \Phi} {{\bf H}} + {{\bf h}}_d^H\right) {\bf w}\right|^2 < (2^R -1){\sigma^2}\right\}.
	\end{align}
	The optimization problem for minimizing the  outage probability in (\ref{out}) under the constraints of the maximum transmit power of AP and unit-modulus reflection of RIS can be formulated as 
	\begin{align}
	\text{(P0):} \quad \underset{{\bf w},{\bm \Phi}}{\min}
	& \quad  t\\
	 \text{s.t.} & \quad \mathbb{P} \left\{ {P_t}
	\left|\left({\bf s}^T {\bm \Phi} {{\bf H}} + {{\bf h}}_d^H\right) {\bf w}\right|^2 < (2^R -1){\sigma^2}\right\}<t \\
	& \quad \left\| {\bf w}\right\|^2 \le 1 \\
	& \quad |\varphi_{g}|=1, \quad g=1,2,\ldots,G  .
	\end{align}
	However, this problem is non-convex and difficult to solve due to the robust outage probability and unit-modulus constraint. 
	On the other hand, it can be observed from (9) that given any ON/OFF-state vector $\bf s$, we should optimize the active beamforming at the AP and passive beamforming at the RIS to maximize the received signal power for minimizing the outage probability. 
	However, since the instantaneous ON/OFF state information of the RIS is unavailable to the AP, we instead optimize the active beamforming and passive beamforming to maximize the
	{\it average} received signal power. Note that this is effective in improving the performance in terms of outage probability, as will be shown in our simulations. 
Let ${\bf A}=\mathbb{E}_{\bf s}\left\{{\bf s} {\bf s}^T\right\}$ and ${\bf a}=\mathbb{E}_{\bf s}\left\{{\bf s} \right\}$  denote the covariance matrix and mean vector of $\bf s$, respectively. 
Based on (\ref{state}) and  assuming that all possible index combinations are equiprobable, the elements of ${\bf A}$ and ${\bf a}$ can be derived as follows.
\begin{align}
\left[{\bf A}\right]_{i,j}=\left\{ \begin{gathered}
\frac{\bar K}{G},\qquad\qquad~ i=j \hfill \\
\frac{{\bar K}({\bar K}-1)}{G({G}-1)},\quad i \neq j \hfill
\end{gathered}  \right. \quad i,j = 1, 2,\ldots,G
\end{align} 
and 
\begin{align}
\left[{\bf a}\right]_{i}=\frac{\bar K}{G}.
\end{align}
Thus, the average received signal power normalized by the transmit power is given by 
	\begin{align}\notag
	&\mathbb{E}_{\bf s}\left\{ 
	\left|\left({\bf s}^T {\bm \Phi} {{\bf H}} + {{\bf h}}_d^H\right) {\bf w}\right|^2\right\}\\\notag
	=&
	{\bf w}^H \left({{\bf H}^H} {\bm \Phi}^H \mathbb{E}_{\bf s}\{ {\bf s}{\bf s}^T \} {\bm \Phi} {{\bf H}} +{{\bf H}^H} {\bm \Phi}^H \mathbb{E}_{\bf s}\{ {\bf s}\} {{\bf h}}_d^H+ {{\bf h}}_d \mathbb{E}_{\bf s}\{ {\bf s}^T \} {\bm \Phi} {{\bf H}} + {{\bf h}}_d{{\bf h}}_d^H\right) {\bf w}\\
	=&  {\bf w}^H \left[ {\begin{array}{*{20}{c}}
		{\bm \Phi} {{\bf H}}  \\
		{{\bf h}}_d^H 
		\end{array}} \right]^H  \underbrace{ \left[ {\begin{array}{*{20}{c}}
			{\bf A}&{\bf a}\\
			{{\bf a}^T}&{1}
			\end{array}} \right] }_{\tilde{{\bf A}}}  \left[ {\begin{array}{*{20}{c}}
		{\bm \Phi} {{\bf H} } \\
		{{\bf h}}_d^H 
		\end{array}} \right] {\bf w} .\label{ava_received_power4}
	\end{align}
	Accordingly, the corresponding optimization problem based on the estimated CSI of ${\hat{\bf H}}$ and ${\hat{\bf h}}_d$ is formulated as: 
\begin{align}
\text{(P1):} \quad
 \underset{{\bf w},{\bm \Phi}}{\max}
& \quad {\bf w}^H\left[ {\begin{array}{*{20}{c}}
	{{\bm \Phi} \hat{{\bf H}}}\\
	{\hat{{\bf h}}_d^H}
	\end{array}} \right]^H\tilde{\bf A}\left[ {\begin{array}{*{20}{c}}
	{ {\bm \Phi} \hat{{\bf H}}}\\
	{\hat{{\bf h}}_d^H}
	\end{array}} \right]{\bf w} \label{P6obj}\\
 \text{s.t.} & \quad \left\| {\bf w}\right\|^2 \le 1 \label{P6con0}\\
&  \quad |\varphi_{g}|=1, \quad g=1,2,\ldots,G  \label{P6con1}.
\end{align}
It can be readily verified that problem (P1) is a non-convex problem  as well, since the 
objective function of (\ref{P6obj}) is non-concave with respect to both ${\bf w}$ and ${\bm \Phi}$, and the constraint in (\ref{P6con1}) is not convex.
Moreover, due to the mutual coupling between
${\bf w}$ and ${\bm \Phi}$ in the objective function of (\ref{P6obj}), problem (P1) becomes even more difficult to solve. To circumvent the above difficulties, 
we develop an alternating optimization algorithm to find a high-quality suboptimal solution to problem (P1) in the following, which iteratively optimizes one of ${\bf w}$ and ${\bm \Phi}$ with the other being fixed at each time for decoupling the original problem.

\subsection{Joint Beamforming Design}

\begin{algorithm}[t]
	\caption{Alternating Optimization Algorithm for Solving Problem (P1)} \label{alg2}
	\hspace*{\algorithmicindent}\textbf{Input}: $\hat{{\bf H}}$, $\hat{{\bf h}}_d$,  threshold $\epsilon$, and the maximum iteration number $I$
	\begin{algorithmic}[1]	
		\STATE Initialize the diagonal phase-shift matrix ${\bm \Phi}^{(1)}:={\bf I}_G$ and set the iteration number $n:=1$ 
		\REPEAT
		\STATE Substitute ${\bm \Phi}^{(n)}$ into (\ref{R}) to get $\tilde{\bf R}^{(n)}$, then find the eigenvector corresponding to the maximum eigenvalue of $\tilde{\bf R}^{(n)}$ to obtain the active beamforming ${\bf w}^{(n)}$ 
		\STATE For given ${\bf w}^{(n)}$, solve problem (P1.4) via convex optimization solver and Gaussian randomization of (\ref{Gaussian randomization}) to obtain the RIS phase-shift matrix ${\bm \Phi}^{(n+1)}$ 
		\STATE Update $n:=n+1$
		\UNTIL The fractional increase of (\ref{P6obj}) is less than {$\epsilon$} or $n>I$
	\end{algorithmic}
	\hspace*{\algorithmicindent} \textbf{Output}: ${\bf w}^*$ and ${\bm \Phi}^*$ 
\end{algorithm}

For any given RIS phase-shift matrix ${\bm \Phi}$, problem (P1) can be rewritten as
%
\begin{align}
\text{(P1.1):} \quad \underset{{\bf w}}{\max}
& \quad {\bf w}^H \tilde{\bf R}  {\bf w}\label{P7obj}\\
 \text{s.t.} & \quad \left\| {\bf w}\right\|^2 \le 1
\end{align}
where 
\begin{align}\label{R}
\tilde{\bf R}=\left[ {\begin{array}{*{20}{c}}
	{ {\bm \Phi} \hat{{\bf H}}}\\
	{\hat{{\bf h}}_d^H}
	\end{array}} \right]^H\tilde{\bf A}\left[ {\begin{array}{*{20}{c}}
	{ {\bm \Phi} \hat{{\bf H}}}\\
	{\hat{{\bf h}}_d^H}
	\end{array}} \right].
\end{align} 
Since $\tilde{\bf R}$ is Hermitian, for any non-zero ${\bf w}$, we have the following inequality
\begin{align}
{\bf w}^H \tilde{\bf R}  {\bf w} \le  \lambda_{\max}( \tilde{\bf R} )  \left\| {\bf w}\right\|^2
\end{align}
where $\lambda_{\max}( \tilde{\bf R} )$ denote the maximum eigenvalue of $\tilde{\bf R}$. Let ${\bf v}_{\max}$ denote the eigenvector corresponding to the maximum eigenvalue of $\tilde{\bf R}$. Then, it can be readily verified that the optimal solution to problem (P1.1) is given by ${\bf w}^*={\bf v}_{\max}/\left\| {\bf v}_{\max} \right\|$.  

Next, we optimize ${\bm \varphi}$ based on the given active beamforming ${\bf w}^*$. 
Specifically, for given ${\bf w}^*$, by letting ${\bf \Lambda}= \text{diag}\left(\hat{{\bf H}}  {\bf w}^*\right)$ and ${{g}}_d= \hat{{\bf h}}_d^H {\bf w}^*$, 
problem (P1) can be rewritten as follows (omitted irrelevant terms  for brevity). 

\begin{align}
\text{(P1.2):} \quad
 \underset{{\bm \varphi}}{\max}
 & \quad {\bm \varphi}^H {\bf \Lambda}^H {\bf A} {\bf \Lambda} {\bm \varphi} + {{g}}_d {\bm \varphi}^H {\bf \Lambda}^H {\bf a} +{{g}}_d^H {\bf a}^T {\bf \Lambda} {\bm \varphi} \label{P2.2obj}\\
 \text{s.t.}  & \quad  |\varphi_{g}|=1, \quad g=1,2,\ldots,G  . \label{P2.2con1}
\end{align}
From (\ref{P2.2obj}) and (\ref{P2.2con1}), we see that problem (P1.2) is a non-convex quadratically constrained quadratic program (QCQP), which can be reformulated as a homogeneous QCQP by introducing an auxiliary variable $t$ \cite{Anthony2007qop}, i.e., 
\begin{align}
\text{(P1.3):} \quad \underset{\tilde{\bm \varphi}}\max
& \quad {\tilde{\bm \varphi}^H }{\bf \Xi} {\tilde{\bm \varphi}}
\label{P2.3obj}\\
 \text{s.t.} & \quad \left|\varphi_{g}\right|=1,  \quad g=1,2,\ldots,G  \label{P2.3con1}
\end{align}
where 
\begin{align}
{\bf \Xi}=\left[ {\begin{array}{*{20}{c}}
	{{\bf \Lambda}^H{\bf A}{\bf \Lambda}}&{{{g}}_d{\bf \Lambda}^H{\bf a}}\\
	{{g}}_d^H {\bf a}^T {\bf \Lambda}&{0}
	\end{array}} \right], \quad {\tilde{\bm \varphi}}=\left[ {\begin{array}{*{20}{c}}
	{{\bm \varphi}}\\
	{t}
	\end{array}} \right].
\end{align}
The objective function of (\ref{P2.3obj}) can be rewritten as ${\tilde{\bm \varphi}^H }{\bf \Xi}{\tilde{\bm \varphi}}=\text{tr}\left({\bf \Xi} {\bf Q}\right)$ with ${\bf Q}={\tilde{\bm \varphi}}{\tilde{\bm \varphi}^H }$. Note that ${\bf Q}$ is a positive semidefinite matrix with $\text{rank}\left({\bf Q}\right)=1$. However, as the rank-one constraint is non-convex, we apply the semidefinite relaxation (SDR) method to relax this constraint and reformulate problem (P1.3) as
\begin{align}
\text{(P1.4):} \quad
 \underset{\bf Q}{\text{max}}
& \quad \text{tr}\left({\bf \Xi} {\bf Q}\right)
\label{P5obj}\\
 \text{s.t.} & \quad \left[{\bf Q}\right]_{g,g} =1,~  g=1,2,\ldots,G+1\\
& \quad {\bf Q}\succeq 0
\end{align}
which is a standard convex semidefinite programming (SDP) problem and can be well solved via existing convex optimization solvers such as CVX \cite{grant2014cvx}. It is worth pointing out that after the relaxation, the optimal solution ${\bf Q}^*$ to problem (P1.4) may not be a rank-one solution. Therefore, we retrieve
$\tilde{\bm \varphi}^*$ from ${\bf Q}^*$ as follows.
\begin{align}\label{Gaussian randomization}
\tilde{\bm \varphi}^*=\left\{ \begin{gathered}
{\bf U} {\bf D}^{1/2}{\bf 1}_{G+1},\quad \text{rank}\left({\bf Q}^*\right)=1 \hfill \\
{\bf U} {\bf D}^{1/2}{\bm \varkappa},\quad\quad\quad\quad \text{rank}\left({\bf Q}^*\right)\neq 1 \hfill
\end{gathered}  \right.
\end{align}
where ${\bf Q}^*={\bf U} {\bf D}{\bf U}^H$ is the eigenvalue decomposition of 
${\bf Q}^*$ and ${\bm \varkappa}\sim \mathcal{CN}( {\bf 0}_{G+1},{\bf I}_{G+1})$ is a random vector.
Finally, the suboptimal solution ${\bm \varphi}^*$ to problem (P1.2) is given by
\begin{align}\label{phase_shift}
{\varphi}_{g}^*=\frac{\left[\tilde{\bm \varphi}^*\right]_g{\Big /}\left[\tilde{\bm \varphi}^*\right]_{G+1}}
{\left|\left[\tilde{\bm \varphi}^*\right]_g{\Big /}\left[\tilde{\bm \varphi}^*\right]_{G+1}\right|},\quad g=1,2,\ldots,G .
\end{align}
The algorithm proceeds by iteratively solving subproblems (P1.1) and (P1.4) in an alternating manner until the convergence criterion is met, i.e., the fractional increase of (\ref{P6obj}) is less than a small positive number $\epsilon$, or the maximum number of iterations has been carried out in practice. 
Algorithm~\ref{alg2} summarizes the above procedures. 
	The convergence of the proposed algorithm can be guaranteed by the fact that the objective value of problem (P1) is non-decreasing over iterations and  upper-bounded by a finite value due to the limited transmit power. 
	On the other hand, problem (P1.1) involving the eigenvalue decomposition of an $N\times N$ matrix can be solved with a complexity of $\mathcal{O}(N^3)$, and the SDP problem (P1.4) can be solved with a worst-case complexity of $\mathcal{O}((G+1)^{4.5})$ \cite{Luo2010Semidefinite}. 
	Given the number of iterations $I$, the total complexity for solving problem (P1)  is thus
	given by $\mathcal{O}\left( (N^3+(G+1)^{4.5})I \right)$. 
After getting the optimized phase-shift vector ${\bm \varphi}^*$, the RIS performs passive beamforming with ${\bm \varphi}^*_{\mathbb I}$ only (while the RIS elements belonging to the remaining $(G-\bar K)$ groups are set to be OFF, i.e., ${\bm \varphi}^*_{\mathbb{G}\setminus \mathbb{I}}$=${\bf 0}_{G-\bar{K}}$) according to the selection of ${\mathbb I}$ by RIS. 

\section{Beamforming Design Based on Instantaneous ON/OFF State Information}\label{beamforming1}
In this section, we characterize the upper bound on the received signal power, which serves to compare the above beamforming design based on the statistical ON/OFF state information of the RIS. 
We assume that the AP knows the ON/OFF state information of the RIS exactly in real time. 
It is worth mentioning that this assumption is similar to that in \cite{Roy2019Beyond}. 
Under this assumption, we formulate an optimization problem to maximize the received signal power by jointly designing the active beamforming at the AP and passive beamforming at the RIS based on instantaneous ON/OFF state information of the RIS. 

Given the indices of the ON-state groups $\mathbb{I}$, the corresponding optimization problem based on the estimated CSI  can be formulated as follows (with $P_t$ omitted for brevity). 
\begin{align}
\text{(P2):} \quad \underset{{\bf w},{\bm \theta}}{\text{max}}
& \quad \left|\left({\bm \theta}^T \hat{\bf H} + \hat{\bf h}_d^H\right) {\bf w}\right|^2\label{P1obj}\\
 \text{s.t.} & \quad \left\| {\bf w}\right\|^2 \le 1,\\
& \quad \left|\theta_g\right|=1, \quad\forall~ g  \in \mathbb{I}\label{P1con1},\\
& \quad \theta_g =0,\quad~\forall~ g  \in \mathbb{G}\setminus \mathbb{I}\label{P1con2}.
\end{align}
By eliminating the constraint of (\ref{P1con2}), problem (P2) is equivalent to 
\begin{align}
\text{(P2.1):} \quad \label{P11obj}
 \underset{{\bf w},{\bm \theta}_{\mathbb I}}{\text{max}}
& \quad \left|\left( {\bm \theta}_{\mathbb I}^T \hat{\bf H}_{\mathbb I} + \hat{\bf h}_d^H\right) {\bf w}\right|^2\\
 \text{s.t.} & \quad \left\| {\bf w}\right\|^2 \le 1,\\
& \quad  \left|\theta_{i_k}\right|=1 ,~ i_k  \in \mathbb{I}, ~ k=1,2,\ldots,{\bar{K}}\label{P2con2}
\end{align}
where ${\bm \theta}_{\mathbb I}$ is the sub-vector consisting of the $\bar K$ entries of ${\bm \theta}$ indexed by ${\mathbb I}$, and $ \hat{\bf H}_{\mathbb I}$ is the sub-matrix consisting of the $\bar K$ rows of $\hat{\bf H}$ indexed by ${\mathbb I}$. 
It can be readily verified that problem (P2.1) is non-convex as well, since the objective function of (\ref{P11obj}) is non-concave with respect to both ${\bf w}$ and ${\bm \theta}_{\mathbb I}$, as well as the constraint of (\ref{P2con2}) is not convex. 
Apparently, problem (P2.1) can be solved suboptimally by leveraging the alternating optimization technique similarly to Algorithm~\ref{alg2}. 
Specifically, based on the alternating optimization technique, 
one of ${\bf w}$ and ${\bm \theta}_{\mathbb I}$ is optimized with the other being fixed in each iteration. 
For given active beamforming $\bf w$, problem (P2.1) can be reformulated as
\begin{align}
\text{(P2.2):} \quad \underset{{\bm \theta}_{\mathbb I}}{\text{max}}
& \quad \left|\sum_{k=1}^{\bar K}{ \theta}_{i_k}\hat{\bar{{\bf h}}}_{i_k}^H{\bf w} + \hat{\bf h}_d^H {\bf w}\right|^2\label{P2.1obj}\\
 \text{s.t.} & \quad  \left|\theta_{i_k}\right|=1 ,~ i_k  \in \mathbb{I}, ~ k=1,2,\ldots,{\bar{K}}. 
\end{align}
By exploiting the triangle inequality, we have
\begin{align}
 \left|\sum_{k=1}^{\bar K}{ \theta}_{i_k}\hat{\bar{{\bf h}}}_{i_k}^H{\bf w} + \hat{\bf h}_d^H {\bf w}\right| \le  \sum_{k=1}^{\bar K}\left| {\theta}_{i_k}\hat{\bar{{\bf h}}}_{i_k}^H{\bf w}\right| + \left|\hat{\bf h}_d^H {\bf w}\right|
\end{align}
with equality if and only if $\angle\left({ \theta}_{i_k}\hat{\bar{{\bf h}}}_{i_k}^H{\bf w}\right)=\angle\left( \hat{\bf h}_d^H {\bf w}\right)$ for all $i_k\in\mathbb{I}$. 
Therefore, the optimal phase shift for the $k$-th ($k=1,2,\ldots,{\bar{K}}$) ON-state group is given by
\begin{align}\label{optimized phi}
\phi_{i_k}^*=\angle\left(\hat{\bar{{\bf h}}}_{i_k}^H{\bf w}\right)-\angle\left( \hat{\bf h}_d^H {\bf w}\right), \quad i_k  \in \mathbb{I}. 
\end{align}
For given ${\bm \theta}_{\mathbb I}^*$ in (\ref{optimized phi}), it can be readily obtained that the optimal active beamforming ${\bf w}^*$ is given by
\begin{align}\label{MRT}
{\bf w}^*={\bf w}_{\text{MRT}} \triangleq \frac{\left({\bm \theta}_{\mathbb I}^T \hat{\bf H}_{\mathbb I} + \hat{\bf h}_d^H\right)^H }{\left\|{\bm \theta}_{\mathbb I}^T \hat{\bf H}_{\mathbb I} + \hat{\bf h}_d^H \right\|}
\end{align}
which is the well-known maximum-ratio transmission (MRT). 
$\bm \theta_{\mathbb I}$ and $\bf w$ are iteratively optimized according to (\ref{optimized phi}) and (\ref{MRT}) in an alternating manner until the convergence criterion is met. 
Note that the above solution is guaranteed to converge since the objective value of problem (P2) is non-decreasing over iterations and  the optimal objective value of problem (P2) is finite. 

\section{Performance Analysis}\label{analysis}

In this section, we investigate the proposed RIS-RPM scheme in terms of the outage probability and achievable rate. 

\subsection{Outage Probability}\label{sec of outage}

For ease of exposition, we assume $N=1$ with ${\bf G}\equiv {\bf g}$ and ${\bf h}_d^H \equiv {h}_d^\dagger$, such that the active beamforming vector $\bf w$ can be dropped. 
\revh{Moreover, we consider the Rician fading channel model for all the channels involved, where each channel coefficient equals the superposition of a determined line-of-sight (LoS) component and a non-LoS component (characterized by a complex Gaussian random variable). 
Let $\kappa_\text{AR}$, $\kappa_\text{Ru}$, and $\kappa_\text{Au}$ denote the Rician factors of the AP-RIS, RIS-user, and AP-user links, respectively. 
In particular, the RIS is generally installed on the walls/ceilings to establish a LoS link with the AP to boost the signal strength in its vicinity, while the user is usually in a relatively rich scattering environment. Therefore, we assume $\kappa_\text{AR}=\infty$, $\kappa_\text{Ru}=0$, and $\kappa_\text{Au}=0$, so that the AP-RIS channel has only a fixed LoS component while the AP-user and RIS-user channels can be well characterized by Rayleigh fading. 
Let $\sigma_g^2$ denote the power in the LoS component of the AP-RIS channel and assume ${\bf h}_r^H \sim \mathcal{CN}({\bf 0}_L, \sigma_h^2{\bf I}_L)$ as well as ${h}_d^\dagger\sim \mathcal{CN}({0}, \sigma_d^2)$.} 
As such, 
we have ${\bf H}\equiv {\bf h}\sim \mathcal{CN}({\bf 0}_G, \sigma_r^2{\bf I}_G)$ with $\sigma_r^2=\bar{L}\sigma_h^2\sigma_g^2$. 
For ease of notation, we assume $\sigma_d^2=\sigma_r^2=1$. 
Moreover, we resort to a unified definition of signal-to-noise ratio (SNR) as $\gamma = P_t/\sigma^2$ to draw essential insights.
Then the outage probability in (\ref{out}) can be simplified as
\begin{align}
p_\text{out} (R) = \mathbb{P} \left\{ 
\left| {\bf s}^T {\bm \Phi} {{\bf h}} + {{h}}_d^\dagger \right|^2 < \frac{2^R -1}{\gamma}\right\}.
\end{align}
Let $\phi_0 = \angle ({ h}_d^\dagger)$ and ${\bm \chi}\triangleq [ \chi_0,\chi_1,\ldots,\chi_G]^T=\left[ {\begin{array}{*{20}{c}}
	{ { h}_d^\dagger }  \\
	{ {\bm \Phi}{\bf h} } 
	\end{array}} \right]$. 
First, we assume perfect CSI available at the AP for setting the phase shifts of RIS elements. 
The optimal phase shifts for $N=1$ are solutions that 
arrange the signals reflected by RIS elements to align in phase with the signal over the direct link at the user, regardless of the ON/OFF-state information of the RIS. 
Thus, we have ${\bm \chi}=e^{j \phi_0 }\left[ |\chi_0|, |\chi_1|,\ldots,|\chi_G| \right]^T$, where $\{ |\chi_g| \}_{g=0}^G$ are i.i.d. Rayleigh random variables with parameter $\sqrt{1/2}$. 
Let $X=\left| {\bf s}^T {\bm \Phi} {{\bf h}} + {{h}}_d^\dagger \right|^2$, which is the square of the sum of $(\bar{K}+1)$ independent Rayleigh random variables and has a Gamma distribution with parameters 
\begin{align}
k_x = \frac{\mathbb{E}\{X\}^2}{ \mathbb{E}\{X^2\} - \mathbb{E}\{X\}^2 }, \quad \theta_x= \frac{ \mathbb{E}\{X^2\} - \mathbb{E}\{X\}^2 }{\mathbb{E}\{X\}}
\end{align}
where 
\begin{align}
\mathbb{E}\{X\}=&(\bar{K}+1) (1+ \frac{{\pi}}{4} \bar{K})\\
\mathbb{E}\{X^2\}=&2 (\bar{K}+1) +( \frac{3 {\pi}}{2} +3 ) (\bar{K}+1)\bar{K} + \frac{3 \pi}{2}(\bar{K}+1)\bar{K}(\bar{K}-1) \notag\\
&+ \frac{\pi^2}{16}(\bar{K}+1)\bar{K}(\bar{K}-1)(\bar{K}-2).
\end{align}	
Its probability density function is
\begin{align}
p_X(x) = x^{k_x-1} \frac{e^{ -x/\theta_x}}{ (\theta_x)^{k_x} \Gamma (k_x)}, \quad x\ge 0 .
\end{align}
Approximating $e^{ -x/\theta_x}$ by $1$ for $x\rightarrow 0 $, we have 
\begin{align}
\mathbb{P} \left\{ X < \delta\right\} \approx \frac{\delta^{k_x} }{ (\theta_x)^{k_x} \Gamma (k_x+1)}
\end{align}
for a small positive value  $\delta \rightarrow 0 $.  
Hence at high SNR the outage probability can be approximated by 
\begin{align}\label{pout}
p_\text{out} (R) \approx p_\text{out}^\text{U} (R) \triangleq\frac{(2^R -1)^{k_x}}{ (\theta_x)^{k_x} \Gamma (k_x+1)}{\gamma^{-k_x}} . 
\end{align}
From (\ref{pout}), we see a diversity gain of $k_x$. 
It can be readily verified that $k_x$ linearly increases  with $\bar{K}$ and satisfies $1 \le k_x\le \bar{K}+1$ with equality if and only if $\bar{K}=0$. 
As shown in Fig.~2, 
$p_\text{out}^\text{U} (R)$ tracks very well the trend of $p_\text{out} (R)$ in the high SNR region, and the increase of $\bar{K}$ leads to an increase in the diversity gain. 

\begin{figure}[!t]
	\centering
	\includegraphics[width=4in]{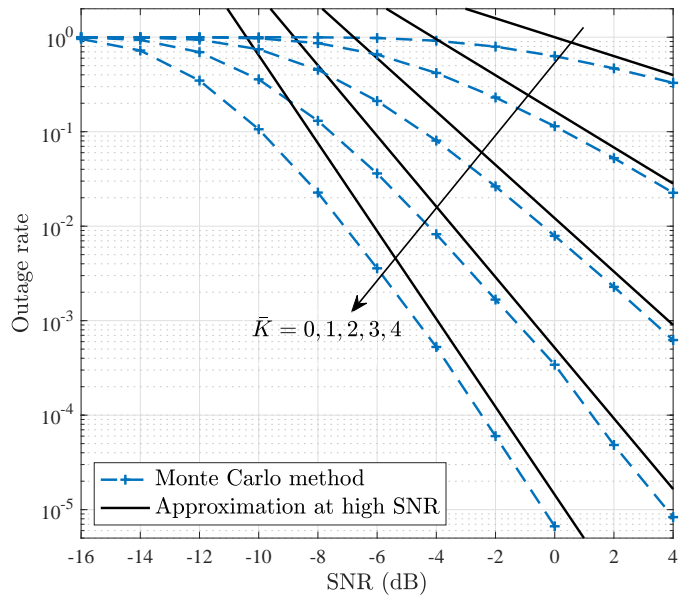}
	\caption{Outage rate of the proposed RIS-RPM scheme in the case of $N=1$ assuming $\sigma_d^2=\sigma_r^2=1$ and perfect CSI available at the AP, where $G=4$, $R=1$ and $\bar{K}$ varies from $0$ to $G$.}
	\label{outage}
\end{figure}

Next, we consider the benchmark case with unit phase shifts at the RIS, i.e., ${\bm \Phi}={\bf I}_G$, which does not require
any CSI for setting the phase shifts and thus can dispense with the channel acquisition. 
In this case, 
$X$ is the square of the sum of $\bar{K}+1$ independent
$\mathcal{CN}(0,1)$ random variables and follows Gamma distribution with parameters $k_x=1$ and $\theta_x=\bar{K}+1$. 
Hence at high SNR the outage probability with the unit phase shift design can be approximated by  
\begin{align}\label{pout1}
p_\text{out}^\text{unit} (R) \approx \frac{2^R -1}{ (\bar{K} +1)\gamma} . 
\end{align}
Comparing (\ref{pout}) with (\ref{pout1}), we can see an increase in the diversity gain due to the properly designed RIS phase shifts.

\subsection{Achievable Rate}
For practical implementation, constellation $\mathbb{A}$ is typically a finite and discrete complex signal set of cardinality $M$ with normalized power, i.e., $\mathbb{A}\triangleq \left\{ a_m \right\}_{m=1}^M$ with $\mathbb{E}\{|a_m|^2\}=1$, where the constellation points are independent and equiprobable. 
Let $\mathbb{S}\triangleq \left\{ {\mathbb{S}_j} \right\}_{j=1}^J$ denote the index set of all the possible combinations of the $\bar{K}$ ON-state groups with the cardinality of $J={G\choose{\bar K}}$, where $\mathbb{S}_j$ is the $j\text{-th}$ element of set $\mathbb{S}$ representing the $j$-th index combination realization. 
Assume that all the combination realizations in $\mathbb{S}$ are independent and equiprobable. 

\subsubsection{Beamforming Design Based on Statistical ON/OFF State Information}

First, we focus on the RIS-RPM scheme with the practical beamforming design based on statistical ON/OFF state information of the RIS. 
We rewrite the system model of (\ref{received2}) as
\begin{align}
y=\sqrt{P_t}\left({\bf s}^T {\bm \Phi} {\bf H} {\bf w} + {\bf h}_d^H {\bf w} \right) x + n .
\end{align}
Let ${\tilde{\bf g}}_r={\bm \Phi} {\bf H}  {\bf w}$, whose entries are the effective channels of the cascaded AP-RIS-user links perceived by the user associated with the corresponding RIS-elements groups. 
Let ${\tilde{g}}_d= {\bf h}_d^H {\bf w}$ denote the effective channel of the direct AP-user link perceived by the user. 
Before recovering the information to be sent by the AP and RIS, the user generally has to acquire the knowledge of ${\tilde{\bf g}}_r$ and ${\tilde{g}}_d$. 
Note that the knowledge of ${\tilde{\bf g}}_r$ can be obtained at the user based on downlink pilot training by using ${\bm \Phi}$ as the RIS reflection pattern, 
while the knowledge of ${\tilde{g}}_d$ can be obtained by applying conventional channel estimation methods with all RIS-elements groups turned OFF. 
Therefore, we assume that the knowledge of $\tilde{\bf g}_r$ and $\tilde{ g}_d$ is available at the user side, and the achievable rate of the RIS-RPM scheme with the practical beamforming design is given by
\begin{align}\label{dif-mi}
\bar{R}_\text{RIS-RPM}=&\mathbb{E}_{\tilde{\bf g}_r,\tilde{ g}_d}\left\{\text{I}\left(\mathbb{I},x;y|\tilde{\bf g}_r,\tilde{ g}_d\right)\right\} 
\end{align}
where $\text{I}\left(X,Y;Z\right)$ denotes the mutual information between the random vector $(X,Y)$ and random variable $Z$. 

\begin{figure*}[!t]
	\begin{align}\label{mi-h case2}
	\bar{R}_\text{RIS-RPM}\hspace{-0.1cm}=&\log_2{J}+\log_2 M -\log_2e \notag\\
	&-\frac{1}{{J}}\frac{1}{M}\sum_{j=1}^{J}\sum_{m=1}^{M}\mathbb{E}_{\tilde{\bf g}_r,\tilde{ g}_d,v}\left\{\log_2 \sum_{j'=1}^{J}\sum_{m'=1}^{M} e^{-\frac{\left| \frac{v}{\sqrt{P_t}}+\left(\sum\limits_{k\in \mathbb{S}_j} \left[{\tilde{\bf g}}_r\right]_k + {\tilde{g}}_d\right)a_m-\left(\sum\limits_{k\in \mathbb{S}_{j'}} \left[{\tilde{\bf g}}_r\right]_k + {\tilde{g}}_d\right)a_{m'} \right|^2}{\sigma^2/P_t}} \right\}
	\end{align}
	\hrulefill
\end{figure*}

\begin{proposition}
	The achievable rate of the RIS-RPM scheme with the practical beamforming design is given by (\ref{mi-h case2}), which is shown at the top of the page, where $v$ is a complex Gaussian random variable following $\mathcal{CN}( { 0},\sigma^2)$. 
\end{proposition}

\begin{proof}
	According to the definition of mutual information, (\ref{dif-mi}) can be derived as
	\begin{align}\label{dif-mi-h case1}
	\bar{R}_\text{RIS-RPM}=\text{H}\left(\mathbb{I},x\right)-\mathbb{E}_{\tilde{\bf g}_r,\tilde{g}_d}\left\{\text{H}\left(\mathbb{I},x|y,\tilde{\bf g}_r,\tilde{ g}_d\right)\right\}
	\end{align}
	where $\text{H}\left(\cdot\right)$ and $\text{H}\left(\cdot| \cdot \right)$ denote the marginal entropy and conditional entropy, respectively. 
	Due to the independence between the selection of ${\mathbb I}$ and symbol modulation at the AP, we have $\text{H}\left(\mathbb{I},x\right)=\log_2{J}+\log_2 M$. 
	The last term at the right hand side of (\ref{dif-mi-h case1}) can be expressed according to the definition of conditional entropy as
	\begin{align}\label{entropy1 case1}
	\text{H}\hspace{-0.05cm}\left(\mathbb{I},x|y,\tilde{\bf g}_r,\tilde{ g}_d\right)
	\hspace{-0.1cm}=&\frac{1}{J}\frac{1}{M}\sum_{j=1}^{J}\hspace{-0.1cm}\sum_{m=1}^{M}\int_y \hspace{-0.1cm} p\left(y|\mathbb{I}=\mathbb{S}_j,x=a_m,\tilde{\bf g}_r,\tilde{ g}_d\right)\notag\\
	&\times\log_2\frac{ p\left(y\right)}{ \frac{1}{{J}}\frac{1}{M} p\left(y|\mathbb{I}=\mathbb{S}_j,x=a_m,\tilde{\bf g}_r,\tilde{ g}_d\right)}dy
	\end{align}
	where 
	\begin{align}\label{pdf1 case1}
	p\left(y|\mathbb{I}=\mathbb{S}_j,x=a_m,\tilde{\bf g}_r,\tilde{ g}_d\right)=\frac{1}{\pi \sigma^2}e^{-\frac{\left|y-\sqrt{P_t} \left(\sum\limits_{k\in \mathbb{S}_j} \left[{\tilde{\bf g}}_r\right]_k + {\tilde{g}}_d\right)a_m\right|^2}{\sigma^2}}
	\end{align}
	and 
	\begin{align}\label{pdf2 case1}
	p\left(y\right)=\frac{1}{{J}}\frac{1}{M}\frac{1}{\pi \sigma^2}\sum_{j'=1}^{J}\sum_{m'=1}^{M}e^{-\frac{\left|y-\sqrt{P_t}\left(\sum\limits_{k\in \mathbb{S}_{j'}} \left[{\tilde{\bf g}}_r\right]_k + {\tilde{g}}_d\right)a_{m'} \right|^2}{\sigma^2}}.
	\end{align}
	Replacing $y$ with $v\triangleq y-\sqrt{P_t} \left(\sum\limits_{k\in \mathbb{S}_j} \left[{\tilde{\bf g}}_r\right]_k + {\tilde{g}}_d\right)a_m$ yields
	\begin{align}\label{part entropy case1}
	&\text{H}\left(\mathbb{I},x|y,\tilde{\bf g}_r,\tilde{ g}_d\right)=-\log_2 \pi\sigma^2+\frac{1}{{J}}\frac{1}{M}\sum_{j=1}^{J}\sum_{m=1}^{M}\int_{v} p(v)\notag\\
	&\times \hspace{-0.1cm}\log_2 \hspace{-0.1cm} \sum_{j'=1}^{J} \hspace{-0.1cm} \sum_{m'=1}^{M} \hspace{-0.15cm} e^{-\frac{\left| \hspace{-0.05cm} \frac{v}{\sqrt{P_t}}+ \hspace{-0.05cm} \left(\sum\limits_{k\in \mathbb{S}_j} \left[{\tilde{\bf g}}_r\right]_k + {\tilde{g}}_d\right)a_m \hspace{-0.05cm} -\hspace{-0.05cm}\left(\sum\limits_{k\in \mathbb{S}_{j'}} \left[{\tilde{\bf g}}_r\right]_k + {\tilde{g}}_d\right)a_{m'} \hspace{-0.05cm} \right|^2}{\sigma^2/P_t}} \hspace{-0.15cm} dv +\frac{1}{{J}}\frac{1}{M}\sum_{i=1}^{J}\sum_{j=1}^{M}\int_{v} p(v)\log_2 \frac{1}{p(v)}  dv 
	\end{align}
	where $p(v)= \frac{1}{\pi \sigma^2}\exp({-\frac{|v|^2}{\sigma^2}})$ is the PDF of a complex Gaussian random variable with zero mean and variance $\sigma^2$. 
	Since the last term at the right hand side of (\ref{part entropy case1}) is the differential entropy of a $\mathcal{CN}( { 0},\sigma^2)$ random variable, which is equal to $\log_2\pi\sigma^2e$, we finally obtain the expression of  $\bar{R}_\text{RIS-RPM}$ as (\ref{mi-h case2}). 
\end{proof}

\subsubsection{Beamforming Design Based on Instantaneous ON/OFF State Information}

Next, we characterize the upper bound on the achievable rate of the RIS-RPM scheme, where the active and passive beamforming vectors are optimized based on the instantaneous ON/OFF state information of the RIS.
In this case, beamforming ${\bf w}$ and ${\bm \theta}_{\mathbb I}$ highly depend on the selection of ${\mathbb I}$. 
Define 
\begin{align}\label{overall channel}
f({\mathbb{I}},{\bf H},{\bf h}_d) \triangleq\left( {\bm \theta}_{\mathbb I}^T {\bf H}_{\mathbb I} + {\bf h}_d^H\right) {\bf w}
\end{align}
which is the effective channel perceived by the user and can be obtained by downlink channel training. 
By assuming that the knowledge of $f({\mathbb{I}},{\bf H},{\bf h}_d) $ is available at the user side, the achievable rate of the  RIS-RPM scheme with the beamforming design presented in Section~\ref{beamforming1} is given by 
\begin{align}\label{dif-mi case1}
\bar{R}_\text{RIS-RPM}^\text{UB}=\mathbb{E}_{{\bf H},{\bf h}_d}\left\{\text{I}\left(\mathbb{I},x;y|{\bf H},{\bf h}_d\right)\right\}.
\end{align}

\begin{proposition}
The achievable rate of the  RIS-RPM scheme with the beamforming design presented  in Section~\ref{beamforming1} is given by (\ref{mi-h case1}), which is shown at the top of the next page, where $v$ is a complex Gaussian random variable following $\mathcal{CN}( { 0},\sigma^2)$. 

\end{proposition}

\begin{figure*}[!t]
	\begin{align}\label{mi-h case1}
	\bar{R}_\text{RIS-RPM}^\text{UB}=&\log_2 J+\log_2 M -\log_2e \notag\\ &-\frac{1}{{J}}\frac{1}{M}\sum_{j=1}^{J}\sum_{m=1}^{M}\mathbb{E}_{{\bf H},{\bf h}_d,v}\left\{\log_2 \sum_{j'=1}^{J}\sum_{m'=1}^{M} e^{-\frac{\left| v+\sqrt{P_t} f(\mathbb{S}_j,{\bf H},{\bf h}_d) a_m-\sqrt{P_t} f(\mathbb{S}_{j'},{\bf H},{\bf h}_d) a_{m'} \right|^2}{\sigma^2}} \right\}
	\end{align}
	\hrulefill
\end{figure*}

\begin{proof}
	The proof is similar to (\ref{mi-h case2}) and omitted for brevity. 
\end{proof}

\section{Simulation Results and Discussions}\label{simulation}


In this section, simulation results are presented to evaluate the performance of our proposed schemes. 
We consider a three dimensional coordinate system, 
where the centers of the AP and RIS are located at $(0,0,0)$ and $(0,d_0,0)$, respectively, and the user is located at $(0,d_y,d_z)$. For the AP, we consider a uniform linear array of $N=4$ antennas with an antenna spacing of half-wavelength, located in $x$-axis. For the RIS, we consider a uniform square array of $L=12\times 12=144$ elements with an element spacing of half-wavelength, deployed in $x$-$z$ plane. 
The fading channel model for all the channels involved is given by
\begin{align}
h=\sqrt{\frac{\kappa\varpi}{\kappa+1}} e^{-j2\pi{d}/{\vartheta}}+\sqrt{\frac{1}{\kappa+1}}\mathcal{CN}( { 0},\varpi)\notag
\end{align}
where $\kappa$ is the Rician factor, $\varpi$ is the path loss, $d$ is the signal propagation distance, and $\vartheta$ is the signal wavelength. We resort to the simplified path loss model, i.e., $\varpi=C_0 d^{-\alpha}$, where $C_0$ is the path loss at a reference distance of 1 meter (m), and $\alpha$ is the path loss exponent. 
\revh{The Rician factors of the AP-user, AP-RIS, and RIS-user links are set as $\kappa_\text{Au}=0$, $\kappa_\text{AR}=\infty$, and $\kappa_\text{Ru}=0$, respectively, as in Section~\ref{sec of outage}. 
The path loss exponents of the AP-user, AP-RIS, and RIS-user links are set as 3.8, 2.2, and 2.4, respectively. The noise power is set as $\sigma^2=-80$ dBm.} 
Other parameters are set as follows: $d_0=50$ m, $d_z=2$ m, $\vartheta=0.1$ m, \revh{$T_c=150$ symbol sampling periods,} $C_0=30$ dB, and quadrature phase-shift keying (QPSK) for symbol modulation at the AP. We consider the Zadoff-Chu sequence as the pilot sequence and set $P_p = 10$ dBm for channel estimation. 
The maximum iteration number in Algorithm~\ref{alg2} is set as $I=5$ \revh{and the threshold is set as $\epsilon=10^{-4}$. 
The values of $G$ and ${\bar K}$ as well as the AP's transmit power level will be specified later to study their effects on the system performance.} 
In the following simulations, the results are obtained by averaging over more than $1000$ independent channel realizations.

\subsection{Performance of Algorithm~1}

\begin{figure}[!t] 
	\centering    
	\subfigure[Perfect CSI at the AP] {
		\label{perfect CSI}     
		\includegraphics[width=4in]{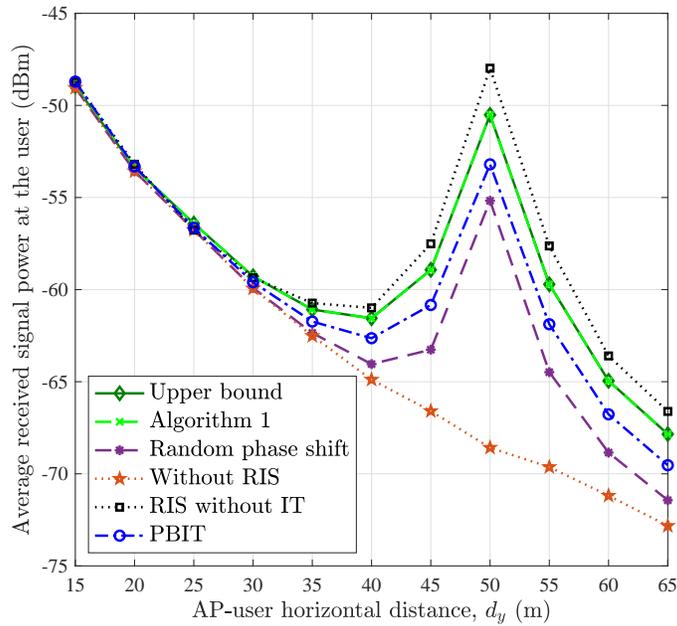}  
	}     
	\subfigure[Estimated CSI at the AP, $P_p=10$ dBm]{ 
		\label{estimated CSI}     
		\includegraphics[width=4in]{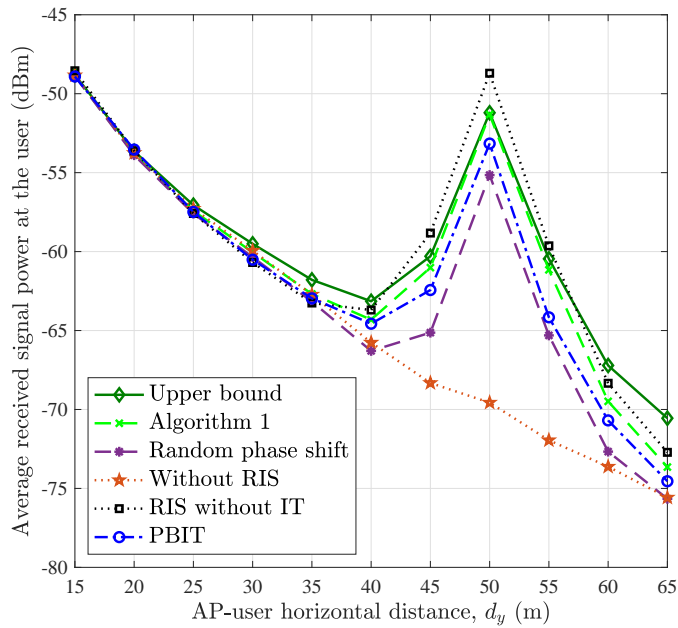}     
	}     
	\revh{\caption{Average received signal power at the user versus AP-user horizontal distance $d_y$, where $G=4$, $P_t=20$ dBm, and ${\bar K}=3$ in Algorithm~\ref{alg2}.}}     
	\label{received pow}     
\end{figure}

To evaluate the effectiveness of Algorithm~\ref{alg2}, we consider the following schemes with fixed $G=4$ and $P_t=20$ dBm: 1) Benchmark scheme without information transfer (IT) \cite{Wu2018Intelligent} where all RIS elements are turned ON and problem (P2) with $\mathbb{I}=\mathbb{G}$ is solved to obtain the active and passive beamforming vectors; 2) Conventional MISO scheme without RIS in which ${\bf w}= {\hat{\bf h}_d}{\big /}{\left\| \hat{\bf h}_d \right\|}$; 3) Random phase shift scheme where the phase shifts of diagonal entries in $\bm \Phi$ are randomly drawn from $(0,2 \pi ]$ and then problem (P1.1) is solved to obtain the active beamforming vector; 4) Upper bound that solves problem (P2); \revh{5) PBIT scheme where elements in the ON/OFF-state vector $\bf s$ are independently drawn from the set $\{0,1\}$ with equal probability. 
For the PBIT scheme, the active and passive beamforming vectors are obtained by solving problem (P1) with ${\bf A}=\frac{1}{4}({\bf 1}_G\cdot{\bf 1}_G^T+{\bf I}_G)$ and ${\bf a}=\frac{1}{2}\cdot {\bf 1}_G$. The number of ON-state groups in Algorithm~\ref{alg2} is set as ${\bar K}=3$.} 
Note that schemes 1) and 2) can be recognized as two special cases of Algorithm~1 with ${\bar K}=G$ and ${\bar K}=0$, respectively. 
Firstly, it can be observed that when the user locates in the vicinity of the RIS, all the RIS-assisted schemes significantly enhance the average received power at the user as compared to the scheme without RIS. 
Secondly, as shown in Fig.~\ref{received pow}, Algorithm~\ref{alg2} significantly outperforms the random phase shift scheme. 
Moreover, as expected in Section~\ref{beamforming1}, the average received  power achieved by Algorithm~1 is upper-bounded by the beamforming design based on the instantaneous ON/OFF state information of the RIS.  
On the other hand, it is worth pointing out that Algorithm~\ref{alg2} incurs small received signal power loss as compared to the scheme without IT. 
This is expected since in Algorithm~\ref{alg2}, one RIS-elements group is turned OFF deliberately to convey additional information of the RIS, while all RIS-elements groups are used for enhancing the reflected signal power in the latter scheme. 
The impact of channel estimation error on the performance of active and passive beamforming is also shown in Fig.~\ref{received pow}. 
We can observe that as compared to the upper bound, Algorithm~\ref{alg2} with perfect CSI at the AP achieves nearly the same performance, while with estimated CSI it performs worse when the user is far away from both the AP and RIS.  
This can be explained by the fact that given the same transmit power at the user, received training signal power is lower when the user moves far away from both the AP and RIS, and thus the channel estimation accuracy is reduced. 

\subsection{Outage Rate Performance}

\begin{figure}[!t]
	\centering
	\includegraphics[width=4in]{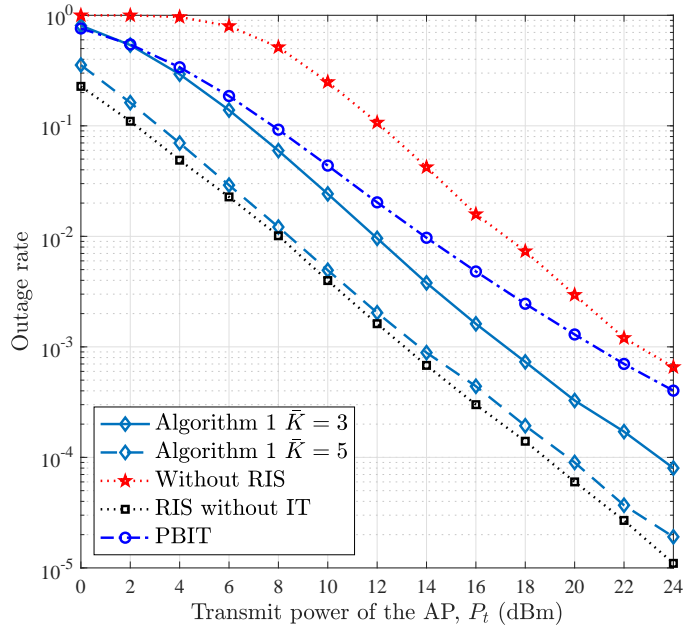}
	\revh{\caption{ Performance comparison between the proposed RIS-RPM scheme and the PBIT counterpart, where $G=6$, $R=1$, and $d_y=45$~m.}}
	\label{outage 1}
\end{figure}

In Fig.~\ref{outage 1}, we consider the PBIT scheme where elements in $\bf s$ are independently drawn from the set $\{0,1\}$ with equal probability. 
Under the PBIT scheme, the average number of ON-state groups at the RIS is $\mathbb{E}\{{\bf s}^T{\bf s}\}=G/2$. 
\revh{The outage rates of our proposed RIS-RPM scheme with ${\bar K}=G-1$ and ${\bar K}=G/2$ are plotted.} 
Two benchmark schemes are considered: 1) Benchmark scheme without IT; 2) Conventional MISO scheme without RIS. 
Fig.~\ref{outage 1} shows the outage rate performance of different schemes versus the AP’s transmit power, where $G=6$, $R=1$ and $d_y=45$ m.  
We can observe that the outage rate performance of the proposed RIS-RPM scheme with ${\bar K}=G/2=3$ outperforms that of the PBIT counterpart. This can be understood by the fact that different from the PBIT counterpart that suffers from large fluctuation in the reflected signal power due to the varying number of ON-state elements, the RIS-RPM scheme keeps the number of ON-state elements at each time constant to reduce such power fluctuation, thus showing better outage rate performance. 
\revh{Moreover, by increasing the number of ON-state groups ${\bar K}$, the outage rate performance of the RIS-RPM scheme can be further improved.} 

\subsection{Achievable Rate Performance}

To evaluate the achievable rate performance of the RIS-RPM scheme, the following schemes are considered: 1) Benchmark scheme without IT; 2) Conventional MISO scheme without RIS; 3) Upper bound that computes the achievable rate of (\ref{mi-h case1})\revh{; 4) PBIT scheme where elements in $\bf s$ are independently drawn from the set $\{0,1\}$ with equal probability.} 
In the sequel, the active and passive beamforming vectors are obtained based on the estimated CSI.

\subsubsection{Effect of AP-User Horizontal Distance $d_y$}

\begin{figure}[!t] 
	\centering    
	\subfigure[$P_t=0$ dBm] {
		\label{dist1}     
		\includegraphics[width=4in]{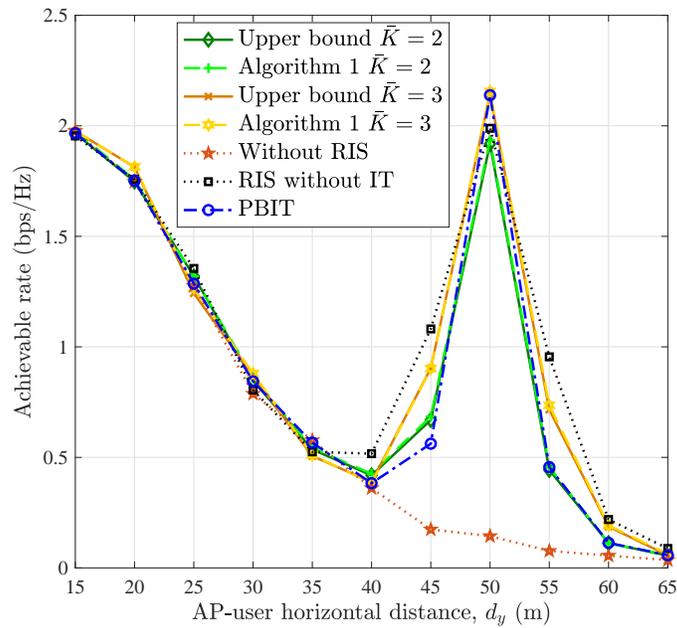}  
	}        
	\subfigure[$P_t=20$ dBm] { 
		\label{dist2}     
		\includegraphics[width=4in]{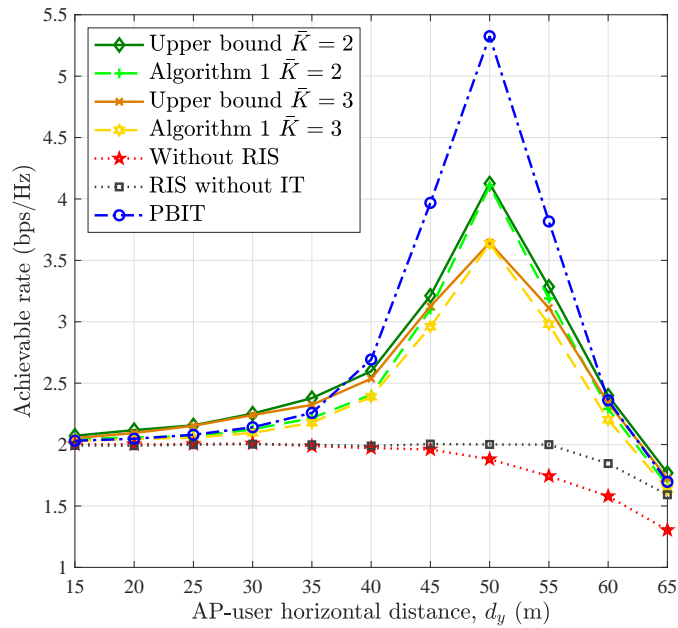}     
	}   
	\revh{\caption{Effect of AP-user horizontal distance $d_y$ on the achievable rate, where $G=4$, and (a) $P_t=0$ dBm; (b) $P_t=20$ dBm.}}     
	\label{dist}     
\end{figure}

In Fig.~\ref{dist}, we evaluate the effect of $d_y$ on the achievable rate, where $G=4$,  ${\bar K}=\{2,3\}$ and $P_t=\{0~\text{dBm}, 20~\text{dBm}\}$. 
One can observe from Fig.~\ref{dist1} that, when the transmit power of the AP is very low, the achievable rate of the scheme without RIS decreases as the user moves away from the AP, and approaches zero. In contrast, the achievable rates of those schemes assisted by the RIS increase drastically as the user moves toward the RIS, and decrease as the user moves away from both the AP and RIS. 
This is because when the user is close to either the AP or RIS, it is able to receive stronger transmitted/reflected signals from the AP/RIS. This phenomenon implies that the cell-edge user can benefit from an RIS deployed in its neighborhood, i.e., the rate of the cell-edge user can be enhanced by deploying an RIS, instead of improving the AP's transmit power or deploying an expensive AP/relay. 
On the other hand, we observe from Fig.~\ref{dist1} that the RIS-RPM scheme with ${\bar K}=3$ has the potential to outperform the one without IT despite that the received signal power achieved by the RIS-RPM scheme is lower than that achieved by the one without IT as shown in Fig.~\ref{estimated CSI}. 
For the RIS-RPM scheme with ${\bar K}=2$, since a large number of RIS elements are turned OFF deliberately for information transfer, the additional information from the RIS cannot compensate for the information reduction caused by the loss of received signal power, thus leading to a smaller achievable rate than the scheme without IT. 
Moreover, it can be observed from Fig.~\ref{dist2} that when the AP's transmit power is high, the RIS-RPM scheme exhibits significantly superior rate performance over the scheme without IT, which is attributed to the additional information delivered by the RIS. 
This implies that the RIS-RPM scheme provides a mechanism to enable a flexible tradeoff between the received signal power and achievable rate performance by varying the number of OFF-state groups at the RIS.  
\revh{Moreover, when the transmit power of the AP is high, the PBIT scheme achieves the maximum achievable rate, as expected.}


\subsubsection{Effect of Number of ON-State RIS-elements Groups $\bar{K}$}

\begin{figure}[!t]
	\centering
	\includegraphics[width=4in]{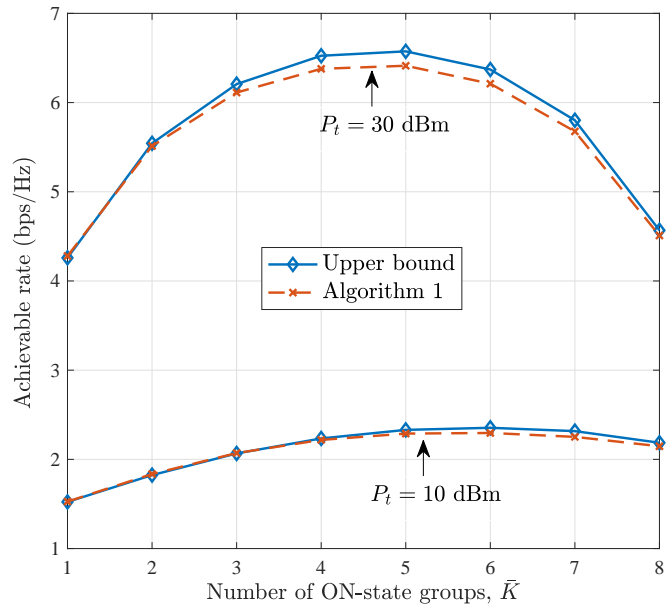}
	\caption{Effect of $\bar{K}$ on the achievable rate, where $G=9$, $d_y=45$ m, and $P_t$ equal to $10$ dBm and $30$ dBm are considered.}
	\label{active}
\end{figure}

We compare the achievable rate of the RIS-RPM scheme versus $\bar{K}$ in Fig.~\ref{active}, where $G=9$, $d_y=45$ m, and $P_t=\{10~\text{dBm}, 30~\text{dBm}\}$ are considered. One can observe that there exists an optimal $\bar{K}$, which varies with different AP's transmit power levels. 
For $P_t=30$ dBm, the optimal $\bar{K}$ is 5 whereas the optimal $\bar{K}$ for $P_t=10$ dBm is 6. 
Generally, it is expected that when the AP's transmit power is high enough, the optimal $\bar{K}$ is more likely to be $\left\lceil G/2\right\rceil$, since the entropy of the RIS is maximized. 
Furthermore, when the AP's transmit power is very high, due to the assumption of finite-alphabet input, the achievable rate at different $\bar{K}$ values will be the sum of the corresponding uncoded transmitted information rates of the AP and RIS. 
In contrast, when the AP's transmit power becomes very low, the optimal $\bar{K}$ is more likely to be $(G-1)$, since the reflected signal power is maximized while the RIS still can convey its information through RPM. 

\subsubsection{Effect of RIS-elements Grouping Ratio}



\begin{figure}[!t]
	\centering
	\includegraphics[width=4in]{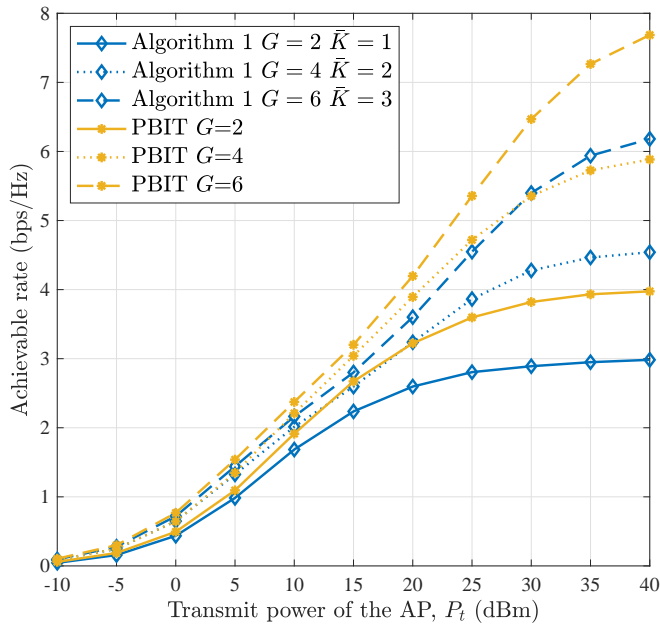}
	\revh{\caption{Effect of RIS-elements grouping ratio on the achievable rate, where $d_y=45$ m and the average number of ON-state RIS elements at each time is $\bar{L}\times\bar{K}=72$.}}
	\label{grouping}
\end{figure}

{\renewcommand{\arraystretch}{1.2}%
	\begin{table}[!t]
		\begin{center}\revh{\caption{Grouping Ratio and Channel Estimation Overhead Ratio}}
			\begin{tabular}{|c|c|c|c|}
				\hline
				$G$ & $2$ & $4$ & $6$  \\
				\hline
				$\rho$ & $1/72$ & $1/36$ & $1/24$ \\
				\hline
				$\xi$ & $3/150$ & $5/150$ & $7/150$ \\
				\hline
				$L_x \times L_z$ & $12 \times 6$ & $6 \times 6$ & $6 \times 4$ \\
				\hline
			\end{tabular}
			\label{grouping ratio}
		\end{center}
	\footnotesize{Each group consists of $\bar{L}=L_x \times L_z$ elements with $L_x$ elements along $x$-axis and $L_z$ elements along $z$-axis.}
\end{table}}

The RIS-elements grouping ratio is defined by $\rho\triangleq 1/\bar{L}$.  
Let $\xi\triangleq (G+1)/T_c$ denote the ratio of time overhead for channel estimation to the coherence time normalized to the symbol sampling period. 
In Fig.~\ref{grouping}, we examine the effect of $\rho$ on the achievable rate, where $d_y=45$~m, $G=\{2,4,6\}$, and $\bar{K}=G/2=\{1,2,3\}$ for the proposed RIS-RPM scheme. 
Note that the average number of ON-state RIS elements keeps constant for different schemes. 
As can be seen, the RIS-RPM scheme with large $G$ can achieve better achievable rate performance. 
This phenomenon can be explained by the fact that with large grouping ratio, not only the degrees of freedom for RIS reflection design increases, achieving high passive beamforming gain, but also more additional information can be conveyed through the index combination of RIS-elements groups, both improving the sum achievable rate. 
However, as shown in Table~\ref{grouping ratio}, the pilot overhead ratio $\xi$ increases with the grouping ratio, which results in more time for channel estimation and less time for data transmission, thus reducing the average achievable rate. 
\revh{On the other hand, we can observe that a large AP's transmit power level is needed for the PBIT scheme to be competitive.} 

\section{Conclusions}\label{conclusion}
In this paper, we considered an RIS-enhanced MISO wireless communication system and proposed the RPM scheme for the dual-use of passive beamforming and information transfer of the RIS. 
A practical beamforming design based on the RIS's statistical ON/OFF state information was proposed to maximize the {\it average} received signal power at the user, 
for which an efficient algorithm based on the alternating optimization technique was proposed to obtain a high-quality solution. 
Next, we formulated an optimization problem to maximize the {\it instantaneous} received signal power by designing active and passive beamforming based on the RIS's instantaneous ON/OFF state information, which characterized the upper bound on the received signal power of the RIS-RPM scheme. 
Moreover, the asymptotic outage probability of the RIS-RPM scheme over Rayleigh fading channels was derived in closed-form. 
In particular, the RIS was shown to be able to increase the diversity gain by properly designing the phase shifts of its elements. 
The achievable rate of the RIS-RPM scheme has been analyzed for the case where the transmitted symbol was drawn from a finite constellation. 
Finally, simulation results corroborated the effectiveness of Algorithm~\ref{alg2} as well as the RIS-RPM scheme and revealed the effect of different system parameters on the achievable rate performance of the RIS-RPM scheme. 
It was shown that the RIS-RPM scheme was able to improve the achievable rate performance despite the loss in received signal power as compared to the conventional RIS-assisted system with full-ON reflection.

\ifCLASSOPTIONcaptionsoff
  \newpage
\fi

\bibliographystyle{IEEEtran}
\bibliography{IRS_inf}

\end{document}